\documentclass[letterpaper, 10 pt, conference]{ieeeconf}
\usepackage[utf8]{inputenc}

\IEEEoverridecommandlockouts

\usepackage{balance}

\usepackage{multirow}

\usepackage{amsmath,mathtools}  %
\usepackage{amsfonts,amssymb}
\usepackage{bbm}

\usepackage{xspace}    %

\usepackage{graphicx}
\graphicspath{{figs/}}
\usepackage{makecell}
\usepackage{subcaption}
\usepackage{caption}

\usepackage{hyperref}
\hypersetup{colorlinks=true, linkcolor= blue, allcolors=blue, citecolor = blue, filecolor=black, urlcolor=blue}

\usepackage{todonotes}
\usepackage{comment}

\usepackage{tikz}
\usetikzlibrary{arrows,positioning}   %

\usepackage{pgfplots}\pgfplotsset{compat=1.9}
\usepackage{grffile}
\pgfplotsset{compat=newest}  %
\usetikzlibrary{plotmarks}
\usetikzlibrary{arrows.meta}
\usepgfplotslibrary{patchplots}
\usepgfplotslibrary{groupplots}

\pgfplotsset{compat=1.15}
\usepackage{mathrsfs}
\usetikzlibrary{arrows}

\usepackage{tabularx}

\usepackage[ruled]{algorithm}
\usepackage[noend]{algpseudocode}

\usepackage{diagbox}
\newcolumntype{K}[1]{>{\centering\arraybackslash}p{#1}}

\makeatletter
\newcommand{\multiline}[1]{%
  \begin{tabularx}{\dimexpr\linewidth-\ALG@thistlm}[t]{@{}X@{}}
    #1
  \end{tabularx}
}
\makeatother

\usepackage[per-mode=symbol]{siunitx}

\usepackage{cite}  %
\usepackage{array}  %
\usepackage{booktabs}           %
\usepackage{balance}

\newcounter{remark}
{\par\endtrivlist\unskip}

\newcounter{problem}
{\par\endtrivlist\unskip}

\usepackage{mylatexdefs}

\newcommand{\bbsym}[1]{\ensuremath{\boldsymbol{#1}}}

\usepackage{stfloats}%

\begin{document}

\title{\LARGE \bf Toward Single-Step MPPI via Differentiable Predictive Control}
\author{Viet-Anh Le$^{1}$, Renukanandan Tumu$^{1}$, and Rahul Mangharam$^{1}$
\thanks{$^{1}$Department of Electrical \& Systems Engineering, University of Pennsylvania, Philadelphia, PA 19104, USA (emails: {\tt\small \{vietanh,nandant,rahulm\}@engineering.upenn.edu}).}
\thanks{This work was partially supported by US DoT Safety21 National University Transportation Center and NSF under Grant CISE-$2431569$.}
}
\maketitle

\begin{abstract}
Model predictive path integral (MPPI) is a sampling-based method for solving complex model predictive control (MPC) problems, but its real-time implementation is challenged by computational and sample requirements that grow with the prediction horizon, as well as sensitivity to manually tuned sampling parameters.
To address these issues, we propose \emph{Step-MPPI}, a framework that learns a sampling distribution and MPPI parameters for efficient single-step lookahead MPPI.
Specifically, a neural network parameterizes the MPPI sampling mean and covariance at each time step, while the single-step cost weights and temperature are jointly learned in a self-supervised manner over long horizons using the MPC cost, constraint penalties, and maximum-entropy regularization.
By embedding long-horizon objectives into the learned cost and sampling policy, Step-MPPI achieves the foresight of multi-step optimization with the millisecond-level latency of single-step lookahead. 
We demonstrate its efficiency across challenging tasks involving high-dimensional systems and/or long control horizons.
\end{abstract}

\section{Introduction}

Model predictive path integral (MPPI) control \cite{williams2017model} is a sampling-based approach for solving nonlinear and nonconvex model predictive control (MPC) problems for both stochastic and deterministic systems \cite{homburger2025optimality}.
By sampling trajectories from a control distribution and updating the control sequence using a cost-weighted strategy, MPPI avoids relying on problem-specific optimizers and is well-suited for nonlinear dynamical systems and complex control tasks.
Its implementation can also be efficiently parallelized, making it attractive for real-time control and widely used in robotics and autonomous driving (see \cite{kazim2024recent} for a survey).
Despite these advantages, MPPI has two important limitations: it can require many samples for long-horizon tasks and is sensitive to the choice of sampling distribution.
Large sample counts can make MPPI computationally expensive on embedded hardware, while the sampling distribution, which is often manually tuned, can strongly affect control performance and robustness.
These limitations pose major challenges for real-time deployment.

Recent studies have focused on improving the MPPI proposal distribution to achieve better performance with fewer rollouts.
Asmar \etal \cite{asmar2023model} proposed adaptive importance sampling to refine the proposal mean and covariance from an initial batch of samples.
Sacks \etal used imitation learning to learn MPPI update rules for low-sample control \cite{sacks_learning_2022}, and later jointly learned the optimizer and warm-starting procedure using model-free reinforcement learning \cite{sacks2024deep}.
Kusumoto \etal \cite{kusumoto2019informed} learned sampling distributions using conditional variational autoencoders trained on optimal solutions.
Belvedere \etal \cite{belvedere2025feedback} augmented MPPI with local linear feedback gains derived from sensitivity analysis.
Qu \etal \cite{qu2023rl} used an offline reinforcement learning policy to initialize MPPI rollouts, while Trevisan \etal \cite{trevisan2024biased} proposed Biased-MPPI using classical or learning-based ancillary controllers to guide sampling.
Crestaz \etal \cite{crestaz2025td} trained a value function via temporal-difference learning to approximate long-term return, enabling shorter-horizon MPPI to retain strong performance.

\begin{figure}
\centering
\includegraphics[width=1.01\linewidth, trim={100 150 70 120}, clip]{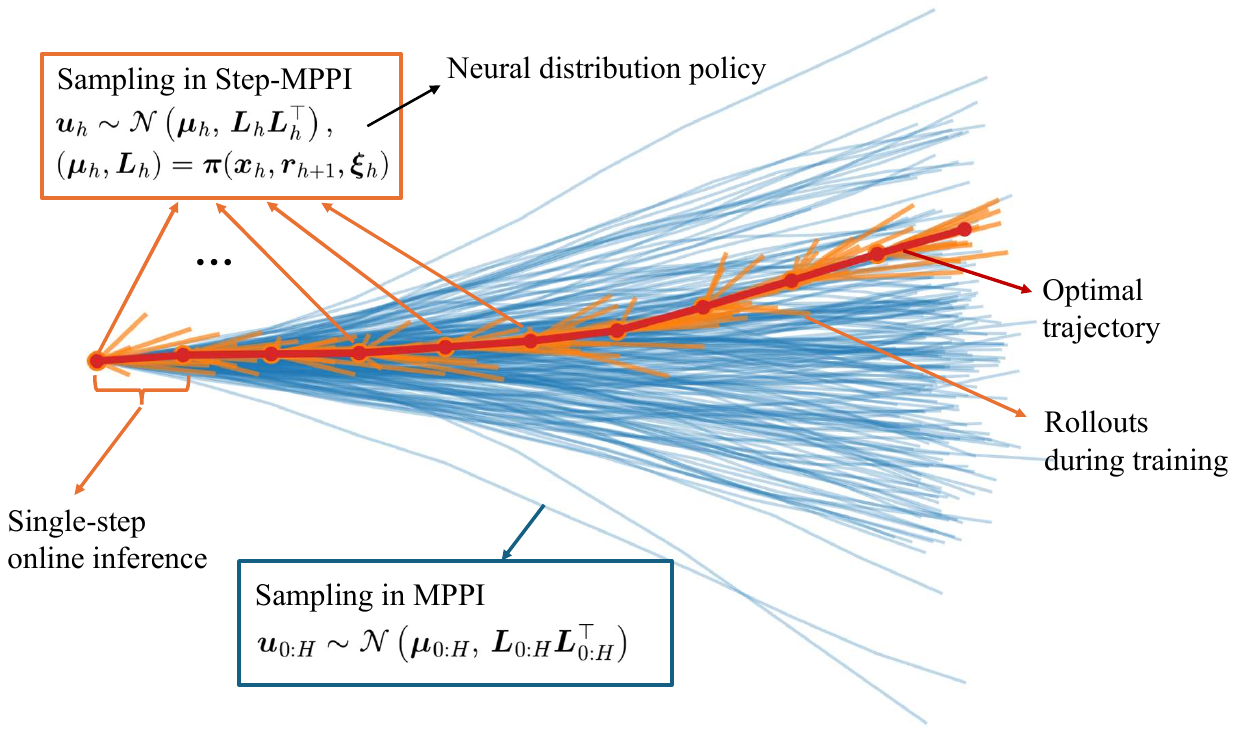}
\caption{Overview of the proposed Step-MPPI versus conventional MPPI. In Step-MPPI, samples are drawn from a distribution parameterized by the neural network, whereas MPPI samples from a nominal distribution over the control horizon. 
During training, Step-MPPI performs rollouts over the full control horizon, while during inference, only a single-step rollout at the current state is needed.
}
\label{fig:overview}
\vspace{-7mm}
\end{figure}

We propose \emph{Step-MPPI}, a framework that learns a \emph{neural distribution policy} to identify the sampling distribution at each time step, together with the weights of a single-step cost and MPPI temperature parameter.
The neural distribution policy parameterizes the mean and covariance of the proposal distribution, which is then used to generate control inputs in a single-step lookahead MPPI, as opposed to learning the distribution over the entire horizon, as done in prior work.
Moreover, the neural distribution policy provides a ``feedback policy'' that guides the sampling distribution toward the optimal solution learned during training, while the learned single-step cost and temperature further adapt the sample evaluation and weighting.
That fundamental discrepancy between Step-MPPI and conventional MPPI can be illustrated in Fig.~\ref{fig:overview}.
The learning procedure is inspired by differentiable predictive control (DPC) \cite{drgovna2024learning}, in which the neural distribution policy, single-step cost, and temperature are trained in a self-supervised manner using a loss function defined by the MPC cost, constraint penalties, and an exploratory regularization term.
Furthermore, to enable end-to-end policy optimization, we formulate the sampling-based MPPI update as a differentiable operator and derive the closed-form Jacobian for the MPPI update layer, bridging the gap between differentiable programming and derivative-free sampling methods.

The proposed framework offers several advantages over existing methods.
First, most prior MPPI methods draw rollout samples over the entire control horizon, but generating useful samples for complex behaviors is difficult without a distributional prior.
By training the distribution policy and the single-step MPPI cost offline, Step-MPPI guides online samples toward low-cost regions.
In addition, learning both the sampling mean and covariance enables a balance between control performance and exploration for improved robustness.
Second, online execution is reduced to a neural network (NN) prediction followed by a single-step MPPI update, resulting in lower computational cost and sample requirements than conventional MPPI.
Finally, unlike deterministic neural policies that may struggle under out-of-distribution conditions, Step-MPPI combines learned proposals with online sampling-based refinement to preserve real-time adaptability.
Compared with DPC, this sampling-based correction layer helps mitigate suboptimality and constraint violations under out-of-distribution conditions.

The remainder of this paper is organized as follows. 
Section~\ref{sec:problem} presents the formulation of the problem and the necessary background, and Section~\ref{sec:approach} details the proposed framework. 
Section~\ref{sec:sim} presents the simulation results, and Section~\ref{sec:con} concludes the paper with final remarks.

\section{Problem Formulation and Preliminaries}
\label{sec:problem}

In this section, we provide the necessary background on sampling-based MPC via MPPI.

\subsection{MPC Problem Formulation}

We consider the control problem for a discrete-time dynamical system with states $\bbsym{x}_t \in \RR^{n_x}$ and control inputs $\bbsym{u}_t \in \RR^{n_u}$. 
The system dynamics are governed by:
\begin{equation}
\bbsym{x}_{t+1} = \bbsym{f} (\bbsym{x}_t, \bbsym{u}_t; \bbsym{\xi}_t),
\label{eq:dynamics}
\end{equation}
where $\bbsym{f}$ is a dynamic model with possibly time-varying model parameters collected in $\bbsym{\xi}_t$. 
In practice, the discrete-time model \eqref{eq:dynamics} is obtained from the continuous-time ODEs using a differentiable discretization method with sampling time $\delta t$.
In MPC, over a control horizon of length $H$, we seek a control sequence $\bbsym{u}_{t:t+H-1}^\top = [\bbsym{u}_t^\top, \bbsym{u}_{t+1}^\top, \cdots, \bbsym{u}_{t+H-1}^\top]$, which results in a state trajectory $\bbsym{x}_{t:t+H}^\top = [\bbsym{x}_t^\top, \bbsym{x}_{t+1}^\top, \cdots, \bbsym{x}_{t+H}^\top]$, that minimizes the following cost function:
\begin{equation}
\begin{multlined}
C(\bbsym{x}_{t:t+H}, \bbsym{u}_{t:t+H-1}; \bbsym{r}_{t:t+H}) = \\ \quad\quad \sum_{h=0}^{H-1} c(\bbsym{x}_{t+h+1}, \bbsym{u}_{t+h}; \bbsym{r}_{t+h+1}),    
\end{multlined}
\label{eq:total_cost}
\end{equation}
where $c (\bbsym{x}_t, \bbsym{u}_t; \bbsym{r}_{t})$ is the stage cost, while $\bbsym{r}_t$ denotes the vector of exogenous information available to the controller, such as future reference signals, time-varying cost parameters, and, when applicable, coefficients associated with state and input constraints.
At each time step $t$, we solve the following optimization problem:
\begin{equation}
\begin{aligned}
\minimize_{\bbsym{u}_{t:t+H-1}} \quad 
& C(\bbsym{x}_{t:t+H}, \bbsym{u}_{t:t+H-1}; \bbsym{r}_{t:t+H}) \\
\subjectto \quad 
& \bbsym{x}_{t+h+1} = \bbsym{f}\!\left(\bbsym{x}_{t+h}, \bbsym{u}_{t+h}; \bbsym{\xi}_{t+h}\right), \\
& h = 0,\dots,H-1.
\end{aligned}
\label{eq:mpc_problem}
\end{equation}
In a receding horizon manner, we apply only the first control input $\bbsym{u}_t^\star$ from the optimizer, move the horizon forward, and resolve \eqref{eq:mpc_problem} at the next time step.

Note that the MPC formulation \eqref{eq:mpc_problem} can accommodate both state and input constraints, which are treated as soft constraints through penalty terms in the cost function. 
Soft constraints are commonly used in both sampling-based MPC and DPC. 
In addition, the formulation in \eqref{eq:total_cost} assigns the same stage cost to the terminal time step as to the preceding time steps. 
Nevertheless, the framework can be readily extended to include an additional terminal cost function, as well as terminal soft constraints.

\subsection{MPPI}

In MPPI \cite{williams2017model}, the optimal control sequence is computed based on multiple samples drawn from a Gaussian distribution. 
Let $\bbsym{z}_h$ denote the parameters of the sampling distribution at time step $h$, which consist of the mean and the Cholesky factor of the covariance matrix. 
Over the control horizon, define
$\bbsym{z}_{t:t+H-1}^\top = [\bbsym{z}_t^\top, \bbsym{z}_{t+1}^\top, \cdots, \bbsym{z}_{t+H-1}^\top]$.
Henceforth, for notation simplicity, we use $\bbsym{x}$, $\bbsym{u}$, $\bbsym{r}$, and $\bbsym{z}$ to denote the corresponding horizon-wise variables. 
Let $\bbsym{\eta}_{\bbsym{z}}$ denote the sampling distribution of the control inputs, \ie $\bbsym{u} \sim \bbsym{\eta}_{\bbsym{z}}$.
Thus, at every time step $t$, the goal of MPPI is to find the over-horizon optimal parameters $\bbsym{z}_{t:t+H-1}$ that solve:
\begin{equation}
\bbsym{z}_{t:t+H-1} = \arg\min_{\bbsym{z} \in \ZZZ} \;  J_t (\bbsym{z}),
\label{eq:mpc_opt}
\end{equation}
where $J_t (\cdot)$ is a statistic defined given the cost function $C(\bbsym{x}, \bbsym{u}; \bbsym{r})$ at time step $t$ which measures the accumulated cost of a trajectory.
In MPPI, we choose the exponential utility for the statistic:
\begin{equation}
J_t (\bbsym{z}) := -\log\, \mathbb{E}_{\eta_{\bbsym{z}}}
\left[ \exp \left(-\frac{1}{\lambda}\, C(\bbsym{x}, \bbsym{u}; \bbsym{r}) \right) \right],
\end{equation}
where $\lambda>0$ is a scaling parameter, also known as temperature. 
We compute the gradients using a likelihood-ratio derivative 
and approximate the expectation with Monte Carlo sampling, which results in the following linear combination of gradients:
\begin{equation}
\nabla J_t (\bbsym{z})
= -\sum_{i=1}^{N} w_i\, \nabla_{\bbsym{z}} \log \bbsym{\eta}_{\bbsym{z}} \big(\bbsym{u}^{(i)}\big),
\label{eq:mppi_grad_mc}
\end{equation}
with weights $w_i$ defined by the \texttt{softmax} operation
\begin{equation}
\label{eq:mppi_weights_softmax}
w_i
=
\frac{
\exp \left(-\frac{1}{\lambda} C(\bbsym{x}^{(i)}, \bbsym{u}^{(i)}; \bbsym{r}) \right)
}{
\sum_{j=1}^{N}
\exp \left(-\frac{1}{\lambda} C(\bbsym{x}^{(j)}, \bbsym{u}^{(j)}; \bbsym{r})\right)
}.
\end{equation}
Next, we choose the Bregman divergence to be the KL divergence and the
distribution to be a factorized Gaussian, \ie 
$\bbsym{\eta}_{\bbsym{z}} (\bbsym{u}) = \prod_{h=0}^{H-1} \NNN (\bbsym{u}_{t+h}; \bbsym{\mu}_{t+h}, \bbsym{\Sigma}_{t+h})$,
for mean vectors $\bbsym{\mu}_{t+h}$ and covariance matrices $\bbsym{\Sigma}_{t+h}$. 
Under these assumptions, we obtain the following update rule for the distribution's mean:
\begin{equation}
\bbsym{\mu}_{t+h}
=
\sum_{i=1}^{N} w_i\, \bbsym{u}_{t+h}^{(i)},
\end{equation}
while the covariance matrix is updated according to
\begin{equation}
\bbsym{\Sigma}_{t+h}
=
\sum_{i=1}^{N} w_i\,\bbsym{m}_{t+h}^{(i)}\big(\bbsym{m}_{t+h}^{(i)}\big)^{\top},
\end{equation}
or be kept fixed over time.
In summary, MPPI samples control sequences over the planning horizon, evaluates their rollout costs, and updates the control sequence using importance-sampling weights, possibly over multiple iterations.
The reader is referred to \cite{williams2017model} for full details on MPPI.

Although MPPI can handle nonlinear dynamics and complex objectives without requiring analytic gradients, its online sampling can become computationally expensive, particularly for long-horizon and high-dimensional control problems.
Moreover, its performance is sensitive to the choice of sampling distribution and temperature parameter.
Inspired by differentiable predictive control (DPC) \cite{drgovna2024learning}, we incorporate offline self-supervised learning into MPPI to reduce the online optimization burden.
The proposed Step-MPPI, presented next, learns a neural sampling distribution, a single-step cost, and the MPPI temperature over long-horizon rollouts, while requiring only single-step MPPI refinement during online execution.

\section{Proposed Approach}
\label{sec:approach}

In this section, we present the details of \emph{Step-MPPI}, which learns a neural sampling distribution, single-step cost weights, and the MPPI temperature to enable single-step MPPI updates while optimizing long-horizon control objectives during training.

\subsection{Single-Step MPPI}
First, we introduce the concept of single-step MPPI, which aims to find $\bbsym{u}_h$ that minimizes the single-step immediate cost function.
In single-step MPPI, suppose that we draw the samples of control inputs at time step $h$ as follows:
\begin{equation}
\label{eq:re-trick}
\bbsym{u}_h = \bbsym{\mu}_h + \bbsym{L}_h \bbsym{\epsilon},
\;
\bbsym{\epsilon} \sim \NNN (\bbsym{0}, \bbsym{I}),    
\end{equation}
where $\bbsym{\mu}_h$ is the mean vector and $\bbsym{L}_h$ is a lower-triangular Cholesky factor of the covariance matrix, while $\bbsym{0}$ and $\bbsym{I}$ are the zero vector and identity matrix with appropriate dimensions. 
\eqref{eq:re-trick} is also known as the \emph{reparameterization trick}.
Suppose that $K$ samples are drawn, and for each sample candidate $\bbsym{u}_h^{(k)}$, $k = 1,\dots, K$, we perform single-step rollouts using the dynamics, and the following MPPI update rule:
\begin{equation}
\label{eq:mppi-update-1}
\bbsym{u}_h^{\mathrm{mppi}}
=
\sum_{k=1}^K w_k \bbsym{u}_h^{(k)}.    
\end{equation}
where the MPPI weights are computed using the learned single-step cost and temperature: 
\begin{equation} 
\label{eq:mppi-weights} 
w_k = \frac{ \exp\!\left( -\frac{1}{\lambda}\, \hat{c}_{\bbsym \phi}\!\left( \bbsym{x}_{h+1}^{(k)}, \bbsym{u}_h^{(k)}; \bbsym{r}_{h+1}^{(k)} \right) \right) }{ \sum_{j=1}^{K} \exp\!\left( -\frac{1}{\lambda}\, \hat{c}_{\bbsym \phi}\!\left( \bbsym{x}_{h+1}^{(j)}, \bbsym{u}_h^{(j)}; \bbsym{r}_{h+1}^{(j)} \right) \right) }. 
\end{equation} 
Here, $\hat{c}_{\bbsym \phi}$ denotes the learned single-step cost parameterized by $\bbsym \phi$, and $\lambda>0$ is a trainable temperature parameter.
In this work, $\hat{c}_{\bbsym{\phi}}$ is parameterized as a linear combination of known features with learnable weights $\bbsym{\phi}$. %
We represent the optimal control input computed by \eqref{eq:mppi-update-1} using the following MPPI layer:
\begin{equation}
\label{eq:dmppi_layer}
\bbsym{u}_{h}^{\mathrm{mppi}} = \mathrm{MPPI}(\bbsym{z}_{h}, \bbsym{s}^{(1:K)}; \bbsym{\phi}, \lambda),
\end{equation}
where $\bbsym{s}^{(1:K)}$ denotes the vector of $K$ cost samples,
and for each $k = 1, \dots, K$, $s^{(k)}$ is computed by: 
\[
s^{(k)} = \hat{c}_{\bbsym \phi} \big( \bbsym{f}(\bbsym{x}_{h}, \bbsym{\mu}_h + \bbsym{L}_h \bbsym{\epsilon}^{(k)}; \bbsym{\xi}_h), \bbsym{\mu}_h + \bbsym{L}_h \bbsym{\epsilon}^{(k)}; \bbsym{r}_{h+1} \big),
\]

Single-step MPPI performs poorly because it solves only a single-step optimal control problem. 
For this reason, the sampling distribution, single-step cost weights, and MPPI temperature are optimized through training so that single-step MPPI can achieve performance comparable to multi-step MPPI.

\subsection{Learning the Step-MPPI Parameters}

We denote $\bbsym{z}_{h} = (\bbsym{\mu}_{h},  \bbsym{L}_{h})$ as the sampling-distribution parameters.
In this work, we propose to learn the optimal parameters of the sampling distribution by an NN as follows:
\begin{equation}
\bbsym{z}_{h} = \bbsym{\pi}_{\bbsym{\theta}} (\bbsym{x}_h, \bbsym{r}_{h+1}, \bbsym{\xi}_{h}).
\end{equation}
The NN weights $\bbsym{\theta}$ are trained alongside with the Step-MPPI parameters $\bbsym{\phi}$ and $\lambda$.
To avoid scale ambiguity and ensure valid parameter values, $\bbsym{\phi}$ is normalized using a \texttt{softmax} parameterization, while $\lambda$ is parameterized using \texttt{softplus} to ensure positivity.
Using only the MPC formulation, as in DPC \cite{drgovna2024learning}, to define the self-supervised loss may cause the covariance of the neural distribution policy to shrink during training. 
To prevent this loss of exploration, we augment the loss function with a maximum-entropy regularization term for the neural distribution policy.
The entropy of a Gaussian distribution $\NNN (\bbsym{\mu}_h, \bbsym{\Sigma}_h)$, denoted by $\HHH \left( \bbsym{z}_h \right)$, is given by \cite{nguyen2025spatially}:
\begin{equation}
\begin{split}
\HHH (\bbsym{z}_h)
&=\frac{1}{2}\log\!\big((2\pi e)^d\, \det \bbsym{\Sigma}_h \big). \\
\end{split}
\end{equation}
Thus, given $M$ training instances $\{ \bbsym{x}_0^{(i)}, \bbsym{r}^{(i)}, \bbsym{\xi}^{(i)} \}_{i=1:M}$, the following loss function is used to jointly train the Step-MPPI parameters:
\begin{subequations}
\label{eq:l2o-mppi}
\begin{align}
\minimize_{\bbsym{\theta}, \bbsym{\phi}, \lambda} \quad 
& \LLL_{\mathrm{policy}} = 
\frac{1}{MH} \sum_{i=1}^{M} \sum_{h=0}^{H-1} \ell_h (\bbsym{\theta}, \bbsym{\phi}, \lambda),
\label{eq:l2o-mppi-a}
\\
\subjectto \quad 
& \bbsym{u}_{h} = \mathrm{MPPI}\big( \bbsym{z}_{h}, \bbsym{s}^{(1:K)}; \bbsym{\phi}, \lambda \big),
\label{eq:l2o-mppi-b}
\\
& \bbsym{z}_{h} = \bbsym{\eta}_{\bbsym{\theta}} (\bbsym{x}_h, \bbsym{r}_{h+1}, \bbsym{\xi}_h),
\label{eq:l2o-mppi-c}
\\
& \bbsym{x}_{h+1} = \bbsym{f} (\bbsym{x}_h, \bbsym{u}_h; \bbsym{\xi}_h),
\label{eq:l2o-mppi-d}
\\
& \bbsym{x}_0 = \bbsym{x}_0^{(i)}, \bbsym{r}_{0:H} = \bbsym{r}^{(i)}, \bbsym{\xi}_{0:H} = \bbsym{\xi}^{(i)},
\label{eq:l2o-mppi-e}
\end{align}
\end{subequations}
for all $h = 0, \dots, H-1$, where the stage loss function $\ell_h (\bbsym{\theta}, \bbsym{\phi}, \lambda)$ is defined by: 
\begin{equation}
\ell_h (\bbsym{\theta}, \bbsym{\phi}, \lambda) := 
c(\bbsym{x}_{h+1}, \bbsym{u}_{h}; \bbsym{r}_{h+1}) - \gamma \HHH ( \bbsym{z}_h)     
\end{equation}
with $\gamma > 0$ being a weighting factor for the regularization term.

Next, we briefly describe the NN architecture used to learn the sampling distribution parameters.
The NN $\bbsym{\pi}_{\bbsym{\theta}}$ uses a shared MLP backbone with two heads that predict the control mean and the entries of the lower-triangular Cholesky factor.
For bounded controls, the mean is squashed using \texttt{tanh}, while the covariance parameters are processed through \texttt{softplus} with clamping for positivity and numerical stability.

\subsection{Differentiation of the Step-MPPI Layer}

To train the parameters of the Step-MPPI layer and the neural distribution policy in an end-to-end manner, we need to differentiate through the Step-MPPI layer \eqref{eq:dmppi_layer}. 
In what follows, we treat the MPPI weighted update as a differentiable layer using the reparameterization trick \eqref{eq:re-trick}, and derive the Jacobian of the layer's outputs with respect to its parameters and inputs.
Note that in this section, we omit time subscripts for simplicity unless needed explicitly.

Using the finite-sample approximation in \eqref{eq:mppi-update-1}, 
since the control samples $\bbsym{u}^{(k)}$ do not directly depend on $\bbsym{\phi}$ or $\lambda$, differentiating the Step-MPPI output with respect to $\bbsym{\phi}$ or $\lambda$ gives
\begin{equation}
\frac{\partial \bbsym{u}^{\mathrm{mppi}}}{\partial \bbsym{\phi}} = \sum_{k=1}^{K} \bbsym{u}^{(k)} \frac{\partial w_k}{\partial \bbsym{\phi}}, \; 
\frac{\partial \bbsym{u}^{\mathrm{mppi}}}{\partial \lambda} = \sum_{k=1}^{K} \bbsym{u}^{(k)} \frac{\partial w_k}{\partial \lambda}.
\end{equation}
Using the gradient of the \texttt{softmax} function yields
\begin{equation}
\frac{\partial w_k}{\partial \bbsym{\phi}}
=
-\frac{1}{\lambda}
w_k
\left(
\bbsym{g}^{(k)}
-
\sum_{j=1}^{K}
w_j\bbsym{g}^{(j)}
\right)^\top,
\end{equation}
where $\bbsym{g}^{(k)} = \nabla_{\bbsym{\phi}} \hat{c}_{\bbsym \phi}\!\left( \bbsym{x}^{(k)}, \bbsym{u}^{(k)}; \bbsym{r} \right)$, and
\begin{equation}
\frac{\partial w_k}{\partial \lambda}
=
\frac{1}{\lambda^2}
w_k
\left(
s^{(k)}
-
\sum_{j=1}^{K}
w_j s^{(j)}
\right).
\end{equation}
Substituting these expressions into the MPPI update directly yields the Jacobians of the Step-MPPI output with respect to $\bbsym{\phi}$ and $\lambda$ as follows,
\begin{align}
\label{eq:diff_u_phi}
\frac{\partial \bbsym{u}^{\mathrm{mppi}}}{\partial \bbsym{\phi}} &= -\frac{1}{\lambda}
\sum_{k=1}^{K} w_k \bbsym{u}^{(k)} \left( \bbsym{g}^{(k)} - \sum_{j=1}^{K} w_j \bbsym{g}^{(j)} \right)^\top, \\
\label{eq:diff_u_lambda}
\frac{\partial \bbsym{u}^{\mathrm{mppi}}}{\partial \lambda} &=
\frac{1}{\lambda^2} \sum_{k=1}^{K} w_k \bbsym{u}^{(k)}
\left( s^{(k)} - \sum_{j=1}^{K} w_j s^{(j)} \right).
\end{align}
The Jacobian of the Step-MPPI layer output with respect to its inputs is obtained by differentiating \eqref{eq:mppi-update-1}, which yields
\begin{equation}
\label{eq:diff_u_z}
\frac{\partial \bbsym{u}}{\partial \bbsym{z}}
=
\sum_{k=1}^K
\left(
\frac{\partial w_k}{\partial \bbsym{z}} \bbsym{u}^{(k)}
+
w_k \frac{\partial \bbsym{u}^{(k)}}{\partial \bbsym{z}}
\right).
\end{equation}
Note that each sample is differentiable with respect to the distribution parameters, and the Jacobians are given by:
\begin{equation}
\label{eq:diff_u0_z}
\frac{\partial \bbsym{u}^{(k)}}{\partial \bbsym{\mu}} = \bbsym{I},
\;\;
\frac{\partial \bbsym{u}^{(k)}}{\partial \bbsym{L}} = \Big(\bbsym{\epsilon}^{(k)}\Big)^\top.
\end{equation}
where with a slight abuse of notation, we let $\partial \bbsym{u}^{(k)} / \partial \bbsym{L}$ denote the Jacobian as a linear operator acting on $d\bbsym{L}$, \ie
$d\bbsym{u}^{(k)} = (d\bbsym{L})\,\bbsym{\epsilon}^{(k)}$.
To complete the differentiation in \eqref{eq:diff_u_z} we need to compute the gradient of $w_k$ with respect to $\bbsym{z}$, which includes $\partial w_k / \partial \bbsym{\mu}$ and $\partial w_k / \partial \bbsym{L}$.
First, by the chain rule,
\begin{equation}
\frac{\partial w_k}{\partial \bbsym{z}} = \sum_{j=1}^K
\frac{\partial w_k}{\partial s^{(j)}} \frac{\partial s^{(j)}}{\partial \bbsym{z}}.
\end{equation}
Using the \texttt{softmax} gradient
\begin{equation}
\frac{\partial w_k}{\partial s^{(j)}} = -\frac{1}{\lambda} w_k(\delta_{kj}-w_j),
\end{equation}
where $\delta_{kj}$ is the Kronecker delta, gives
\begin{equation}
\label{eq:diff_w_z}
\frac{\partial w_k}{\partial \bbsym{z}} =
-\frac{1}{\lambda} w_k
\left( \frac{\partial s^{(k)}}{\partial \bbsym{z}} -
\sum_{j=1}^K w_j\frac{\partial s^{(j)}}{\partial \bbsym{z}} \right).
\end{equation}
Using \eqref{eq:diff_u0_z} we have
\begin{equation}
\frac{\partial s^{(k)}}{\partial \bbsym{\mu}} = \frac{\partial s^{(k)}}{\partial \bbsym{u}^{(k)}}, \;
\frac{\partial s^{(k)}}{\partial \bbsym{L}} = \frac{\partial s^{(k)}}{\partial \bbsym{u}^{(k)}}
\left(\bbsym{\epsilon}^{(k)}\right)^\top.
\end{equation}
Substituting these expressions into \eqref{eq:diff_w_z} yields the gradients of $w_k$ with respect to $\bbsym{\mu}$ and $\bbsym{L}$ as follows
\begin{equation}
\small{
\label{eq:diff_w_mu_L}
\begin{split} \frac{\partial w_k}{\partial \bbsym{\mu}} &= \!
-\frac{1}{\lambda} w_k \left( \frac{\partial s^{(k)}}{\partial \bbsym{u}^{(k)}} - \sum_{j=1}^K w_j \frac{\partial s^{(j)}}{\partial \bbsym{u}^{(j)}} \right), \\
\!\! \frac{\partial w_k}{\partial \bbsym{L}} &= \!
-\frac{1}{\lambda} w_k \Bigg( \frac{\partial s^{(k)}}{\partial \bbsym{u}^{(k)}} \left(\bbsym{\epsilon}^{(k)}\right)^\top \!\! - \sum_{j=1}^K
w_j \frac{\partial s^{(j)}}{\partial \bbsym{u}^{(j)}} \left(\bbsym{\epsilon}^{(j)}\right)^\top \!\! \Bigg).
\end{split}}
\end{equation}

Thus, the Step-MPPI layer is differentiable with respect to its inputs and parameters and can be integrated into end-to-end training.

\section{Results and Discussions}
\label{sec:sim}

In this section, we validate the performance of the proposed framework through three numerical examples: (1) a high-speed autonomous vehicle, (2) a quadrupedal robot, and (3) an urban traffic network (see Fig.~\ref{fig:example}).
While the autonomous vehicle example demonstrates the advantages of the proposed frameworks in the long-horizon MPC problem, the quadrupedal robot example is better suited for validation in a problem involving higher-dimensional state and input spaces, as well as a system that is more sensitive to noise.
In addition, the urban traffic network problem is challenging due to two reasons: the high dimensionality and the long control horizon.
The implementation and simulation videos are available at the supplemental website: \url{https://sites.google.com/seas.upenn.edu/step-mppi}.

\begin{figure*}[t]
\centering
\begin{subfigure}[t]{0.33\textwidth}
    \centering
    \includegraphics[height=3.2cm, trim={50 200 0 240}, clip]{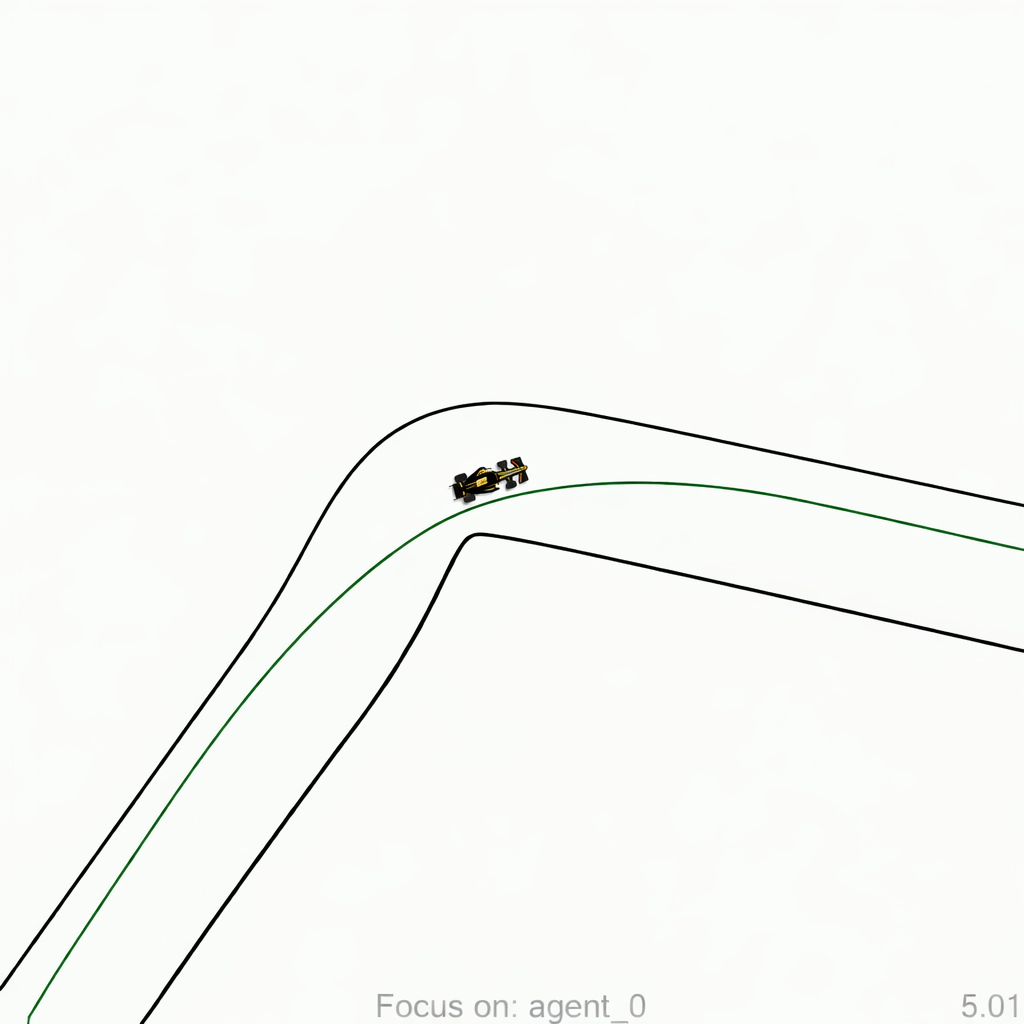}
    \caption{Autonomous vehicle in F1TENTH-Gym}
    \label{fig:ex1}
\end{subfigure}
\hfill
\begin{subfigure}[t]{0.33\textwidth}
    \centering
    \includegraphics[height=3.2cm, trim={800 0 800 0}]{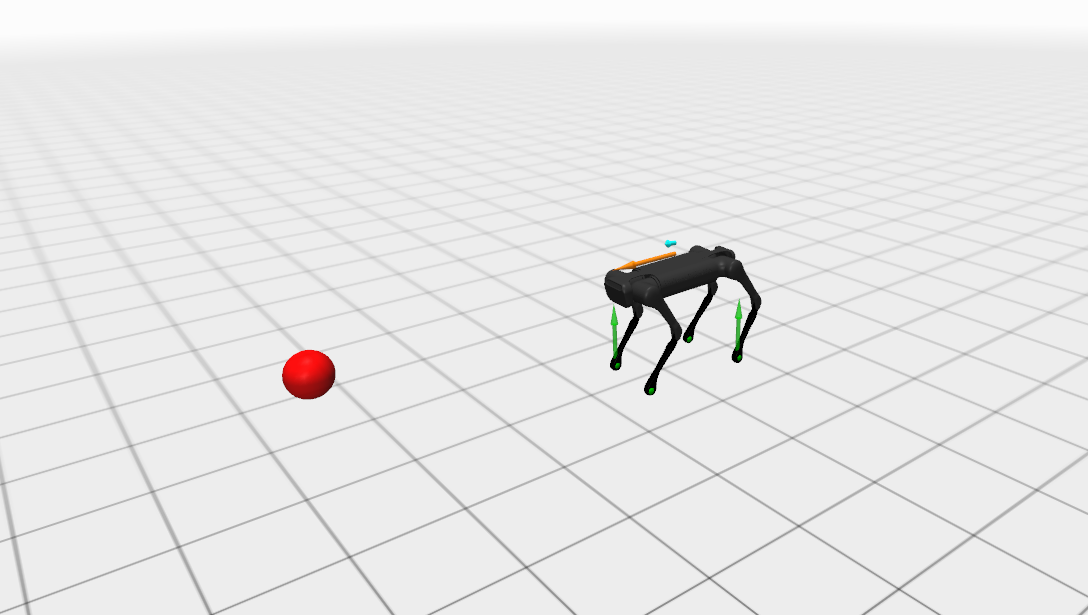}
    \caption{Quadrupedal robot in MuJoCo}
    \label{fig:ex2}
\end{subfigure}
\hfill
\begin{subfigure}[t]{0.32\textwidth}
    \centering
    \includegraphics[height=3.2cm]{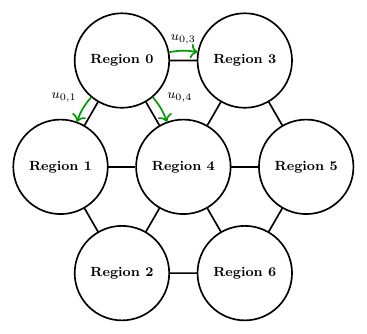}
    \caption{Traffic network with perimeter control}
    \label{fig:ex3}
\end{subfigure}
\caption{Three numerical examples considered for validation.}
\label{fig:example}
\vspace{-4mm}
\end{figure*}

\subsection{High-speed Autonomous Vehicle}

\begin{figure}[t]
\centering
\begin{subfigure}[t]{0.85\linewidth}
\centering
\includegraphics[width=\linewidth]{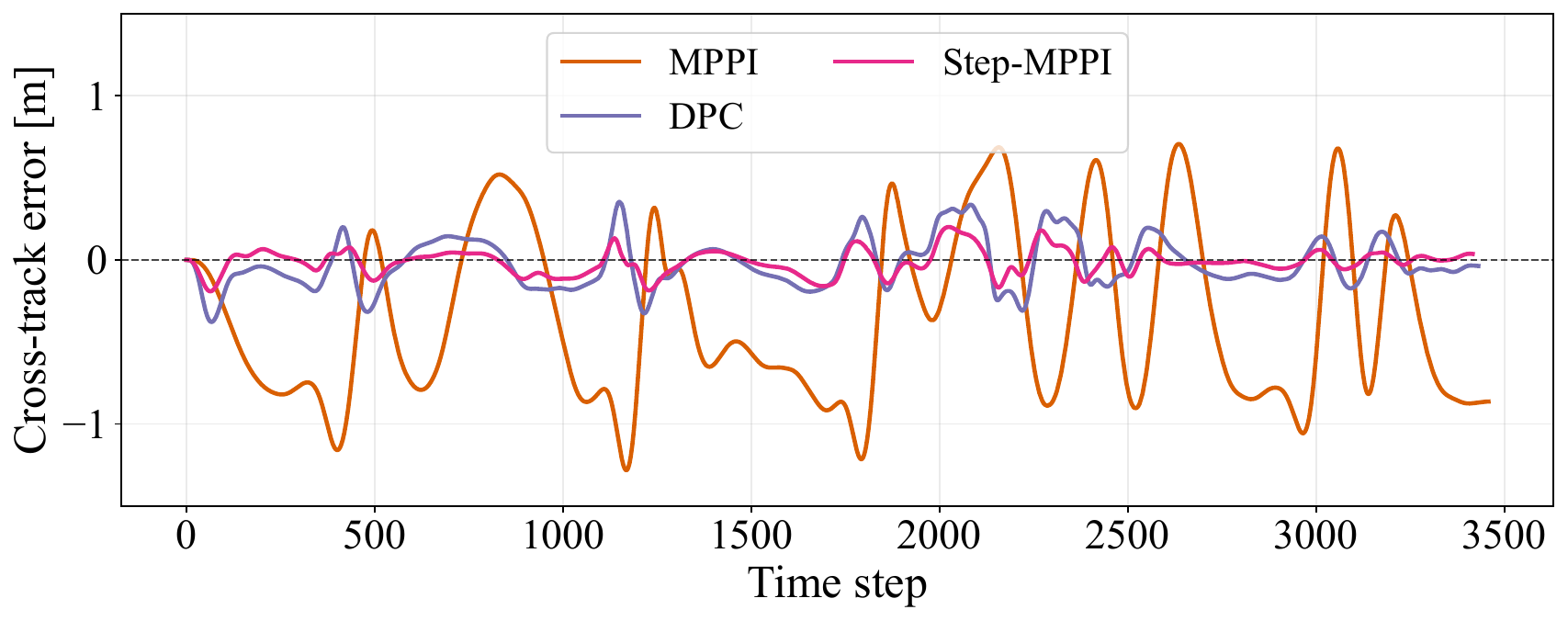}
\caption{Cross-track error over time.}
\label{fig:cte_timeseries}
\end{subfigure}
\hfill
\begin{subfigure}[t]{0.85\linewidth}
\centering
\includegraphics[width=\linewidth]{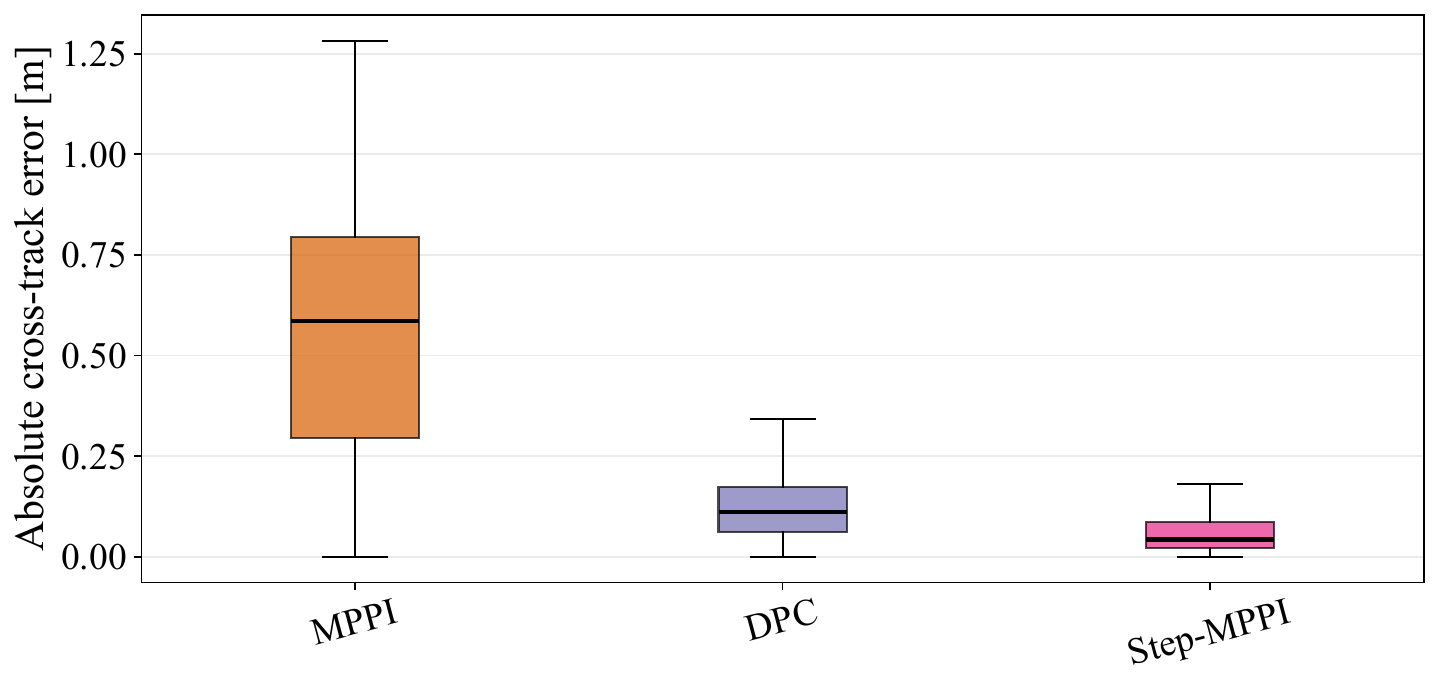}
\caption{Absolute cross-track error distribution.}
\label{fig:cte_boxplot}
\end{subfigure}
\caption{Cross-track error comparison of MPPI, DPC, and Step-MPPI in the autonomous vehicle example.}
\label{fig:cte}
\vspace{-6mm}
\end{figure}

We validate the framework using an autonomous vehicle example in the \texttt{F1TENTH-Gym} environment \cite{o2020f1tenth}.
In this example, the autonomous vehicle tracks a reference trajectory while staying within the track boundaries.
Our empirical results show that, in constrained high-speed driving, the MPC horizon must be sufficiently long for the vehicle to successfully navigate challenging regions, such as sharp turns.
We consider the dynamic single-track model \cite{althoff2017commonroad}, in which the state of the system is given by $\bbsym{x} = [p_x, p_y, v, \delta, \Psi, \dot{\Psi}, \beta]^{\top}$, where $(p_x, p_y)$ is the position in cartesian coordinates, $v$ is the velocity, $\delta$ is the steering angle of the front wheels, $\Psi$ is the heading angle, $\dot{\Psi}$ is the yaw rate, and $\beta$ is the side-slip angle. 
The control inputs $\bbsym{u} = [a, d \delta]^{\top}$ include the acceleration in the longitudinal direction and the steering speed.
The objective function at each time step is to minimize the tracking errors to reference waypoints and the control efforts as follow:
\begin{equation}
\begin{multlined}
c (\bbsym{x}_{t+h+1}, \bbsym{u}_{t+h}; \bbsym{r}_{t+h+1}) \!=\! 
\Big( \!
\left\lVert \bbsym{x}_{t+h+1} \!- \bbsym{r}_{t+h+1} \right\rVert_{\bbsym{Q}}^{2} \\
+
\left\lVert \bbsym{u}_{t+h} \right\rVert_{\bbsym{R}}^{2}
\! \Big),
\end{multlined}
\label{eq:mpc_objective}
\end{equation}
where $\bbsym{Q}$ and $\bbsym{R}$ are the weight matrices with appropriate dimensions.
Safety constraints to the track border are formulated as linear constraints at each time step as follows:
\begin{equation}
c_{\mathrm{right},\,t+h} \le a_{t+h}\, p_{x,t+h} + b_{t+h}\, p_{y,t+h} \le c_{\mathrm{left},\,t+h},
\label{eq:border_constraints}
\end{equation}
where $(x_{t+h},y_{t+h})$ denotes the vehicle position in global coordinates, while $(a_{t+h},b_{t+h},c_{\mathrm{left},\,t+h},c_{\mathrm{right},\,t+h})$ are the coefficients defining the local linear half-space approximation of the left and right track borders.
In addition to the safety constraints, we impose the following bound constraints on the states and inputs:
\begin{align}
& v_{\min} \le v_{t+h} \le v_{\max},\; \delta_{\min} \le \delta_{t+h} \le \delta_{\max} , \\
& a_{\min} \le a_{t+h} \le a_{\max},\; d \delta_{\min} \le d \delta_{t+h} \le d \delta_{\max},
\end{align}
where $v_{\min}, v_{\max}$, $\delta_{\min}, \delta_{\max}$, $a_{\min}, a_{\max}$, and $d\delta_{\min}, d\delta_{\max}$ denote the bounds on speed, steering angle, acceleration, and steering rate, respectively.

We consider the MPC problem with a horizon of length $H = 40$ and sampling time $\delta_t = \SI{0.05}{s}$.
Step-MPPI employs the importance-sampling updates with $1024$ samples. 
Note that Step-MPPI performs only a single iteration at each time step.
We compare them with two baselines: (1) the \emph{DPC} policy, and (2) \emph{MPPI} executing one iterations and $16,384$ rollout samples.
The MPPI implementation with an adaptive penalty method for better constraint handling, as presented in \cite{zang2025sit}, is adopted.
The MPPI hyperparameters have been manually tuned to ensure that the vehicle can complete the lap without crashes.

Figure~\ref{fig:cte} compares the cross-track error of MPPI, DPC, and Step-MPPI in a simulation with a vehicle speed of \SI{10}{m/s}.
In the top panel, we observe that MPPI exhibits the largest oscillations and error magnitudes over the trajectory, while DPC significantly reduces the error but still shows noticeable fluctuations in several intervals. 
In contrast, Step-MPPI maintains smaller cross-track errors and smoother trajectories.
The bottom panel further confirms this observation through the distribution of the absolute cross-track error. 
Overall, Step-MPPI achieves lower median errors with tighter distributions than DPC and MPPI. 
Overall, the figure shows that the proposed method improves tracking accuracy compared with MPPI and DPC.

\subsection{Quadrupedal Robot}

We validate the proposed framework on a quadrupedal robot walking task, whose MPC formulation is adopted from \cite{turrisi2024benefits}. 
The MPC formulation is based on single rigid body dynamics (SRBD), which capture the robot's center-of-mass motion and body rotation while neglecting the swing-leg dynamics. 
Let $\bbsym{p}_c \in \RRR^3$ and $\bbsym{v}_c \in \RRR^3$ denote the center-of-mass position and velocity, respectively; let $\bbsym{\Phi} \in \RRR^3$ denote the orientation, parameterized by roll, pitch, and yaw angles; and let $\prescript{}{\CCC}{\bbsym{\omega}} \in \RRR^3$ denote the angular velocity expressed in the body frame.
Moreover, let $\bbsym{\Gamma}_i \in \RRR^3$ denote the ground reaction forces (GRFs) at each foot $i\in\{1,\dots,4\}$, and $\bbsym{\delta}$ be the contact states of all four legs, which can be computed by a gait optimization layer.
The full details of the SRBD model can be found in \cite{turrisi2024benefits}.
The vectors of states and control inputs are defined as
\begin{equation}
\bbsym{x}^\top = \left[ \bbsym{p}_c^\top,\; \bbsym{v}_c^\top,\; \bbsym{\Phi}^\top,\; \prescript{}{\CCC}{\bbsym{\omega}} \right],
\quad \bbsym{u}^\top = [\bbsym{\Gamma}_i^\top]_{i = 1, \dots, 4}.
\end{equation}
The following quadratic cost function is used for this control task:
\begin{equation}
\begin{multlined}
c (\bbsym{x}_{t+h+1}, \bbsym{u}_{t+h}; \bbsym{r}_{t+h+1}) \!=\! 
\norm{\bbsym{x}_{t+h+1} - \bbsym{x}_{\mathrm{ref}}}_{\bbsym{Q}}^2  + \\ 
\norm{\bbsym{u}_{t+h} - \bbsym{u}_{\mathrm{ref}}}_{\bbsym{R}}^2,    
\end{multlined}
\end{equation}
where $\bbsym{Q}$ and $\bbsym{R}$ are positive diagonal weighting matrices, $\bbsym{x}_{\mathrm{ref}}$ and $\bbsym{u}_{\mathrm{ref}}$ are the desired states and nominal contact forces, respectively, and $\bbsym{r}_{t+h+1}^\top = [\bbsym{x}_{\mathrm{ref}}^\top, \bbsym{u}_{\mathrm{ref}}^\top]$.

\begin{table}[t]
\caption{Comparison of MPPI, DPC, and Step-MPPI on the quadrupedal robot task over $100$ testing simulations with different goal positions.}
\label{tab:quad_summary}
\centering
\begin{tabular*}{\linewidth}{@{\extracolsep{\fill}}lccc@{}}
\toprule
Method & Success & Failure & Timeout \\
\midrule
MPPI & $83$ & $17$ & $0$ \\
DPC & $95$ & $4$ & $1$ \\
Step-MPPI & $\mathbf{100}$ & $\mathbf{0}$ & $\mathbf{0}$ \\
\bottomrule
\end{tabular*}
\vspace{-6mm}
\end{table}

\begin{figure}[t]
\centering
\includegraphics[width=0.8\linewidth]{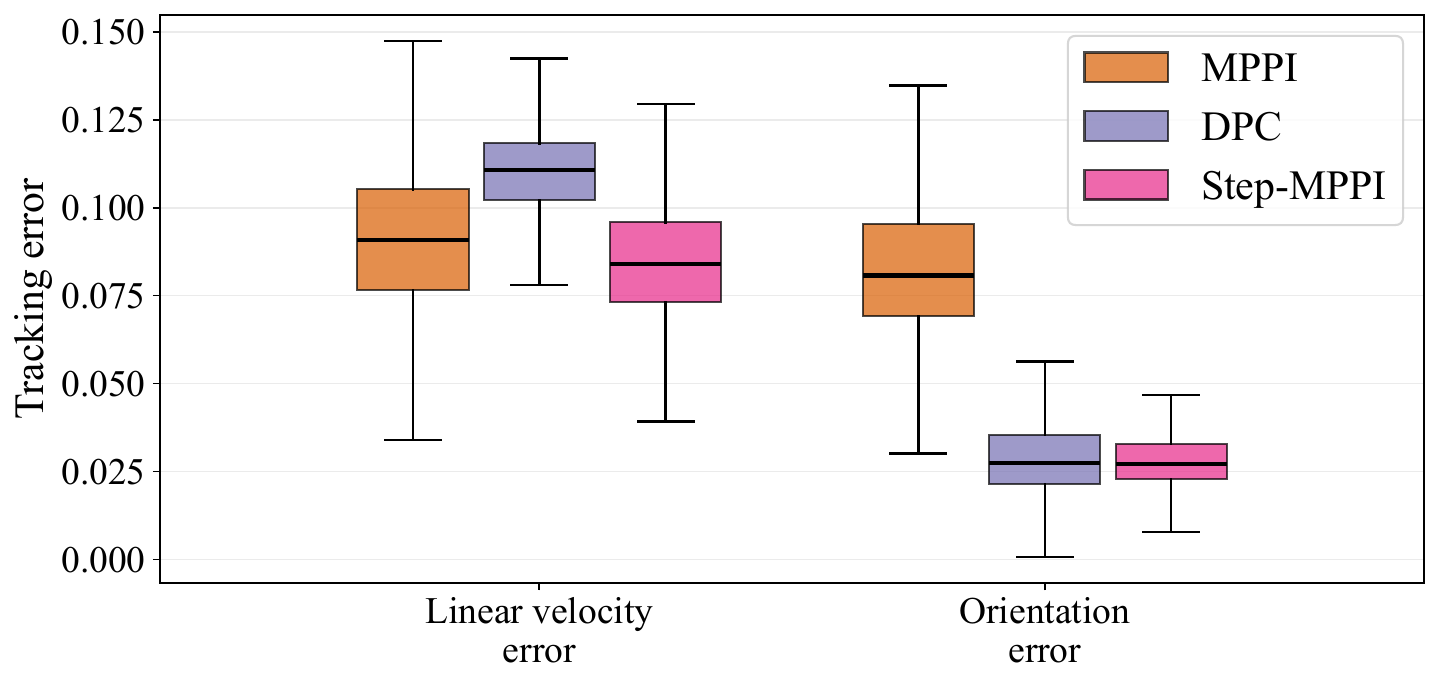}
\caption{Box-plot comparison of tracking performance for MPPI, DPC, and Step-MPPI in the quadrupedal robot task.}
\label{fig:cte_boxplot}
\vspace{-6mm}
\end{figure}

We consider the MPC problem with a horizon of length $H = 20$ and sampling time $\delta_t = \SI{0.02}{s}$.
We compare Step-MPPI with $512$ samples, DPC, and MPPI with $16,384$ rollout samples. 
The performance of these methods is evaluated over $100$ testing simulations with different goal positions.
The results in Table~\ref{tab:quad_summary} show that conventional MPPI, even with a large number of samples, achieves the lowest success rate of $84\%$. 
Meanwhile, Step-MPPI attains the best performance among all tested methods, successfully completing all simulations without any failures or timeouts. 
DPC performs reasonably well, with success rates of $95\%$, indicating that they are less reliable than Step-MPPI in this example.
We additionally present the tracking errors as box plots, including the linear velocity and orientation errors. 
Step-MPPI achieves the lowest linear-velocity error, with MPPI close behind, while DPC exhibits the largest velocity-error distribution. 
For orientation tracking, both DPC and Step-MPPI show much tighter distributions than MPPI, with Step-MPPI achieving the smallest spread and a slightly lower median error. 

\subsection{Urban Traffic Network}

In the urban traffic network setting \cite{tumu_differentiable_2024}, the task is to regulate the flow of traffic through a traffic network in order to dissipate the traffic as quickly as possible. 
The traffic network consists of $R$ regions with indices $\mathcal{R} = \{1,\ldots,R\}$. 
We denote the neighbors of region $i$ as $\mathcal{N}_i \subset \mathcal{R}$.
The state $x_{ij} \geq 0$ is the accumulation in region $i$ with destination $j$, and $\boldsymbol{x}$ is the matrix containing all state components. 
The total accumulation and outflow of region $i$ are computed by:
\begin{equation}
    x_i = \sum_{j\in \mathcal{R}} x_{ij}, \quad
    g_i(x_i) = a_i x_i^3 + b_i x_i^2 + c_i x_i, \label{eq:MFD}
\end{equation}
where $g_i$ is the macroscopic fundamental diagram (MFD), typically a cubic polynomial fit to empirical data \cite{johari_macroscopic_2021}. 
The transfer and completion flows are
\begin{equation}\label{eq:route-flows}
    m_{ihj} = \theta_{ihj} \frac{x_{ij}}{x_i}\,g_i(x_i), \quad
    m_{ii} = \frac{x_{ii}}{x_i}\,g_i(x_i),
\end{equation}
where $\theta_{ihj} \in [0,1]$ is the fraction of traffic in region $i$ with destination $j$ that routes through neighbor $h$, with $\sum_h \theta_{ihj} = 1$. 
The exogenous demand $d_{ij}(t)$ gives the rate at which vehicles appear in region $i$ with destination $j$. 
The perimeter control $u_{ij} \in [\underline{u},\bar{u}]$ regulates flow between adjacent regions. 
The system dynamics are therefore given by:
\begin{subequations}\label{eq:nmfd-update}
\begin{align}
    \dot{x}_{ii} &= d_{ii} - m_{ii}(\boldsymbol{x}) + \sum_{h \in \mathcal{N}_i} u_{hi}\, m_{hii}(\boldsymbol{x}), \\
    \dot{x}_{ij} &= d_{ij} - \sum_{h \in \mathcal{N}_i} u_{ih}\, m_{ihj}(\boldsymbol{x}) + \sum_{\substack{h \in \mathcal{N}_i \\ h \neq j}} u_{hi}\, m_{hij}(\boldsymbol{x}).
\end{align}
\end{subequations}

\begin{table}[!bt]
    \caption{Urban Traffic Network Results.\\{\small{TVH = Total Vehicle Hours, defined as $\sum_k \|\boldsymbol{x}_k\|_1 \cdot \Delta t / 3600$.}}
    \label{tab:results}}
    \centering
    \scriptsize{
    \begin{tabular*}{\linewidth}{@{\extracolsep{\fill}}lccc@{}}
        \toprule
        Method & Final Accumulation (veh) & TVH (veh$\cdot$h) \\
        \midrule
        \multicolumn{3}{l}{\textbf{In-Distribution}} \\
        \midrule
        Baseline   & $10,328.34 \pm 205.82$ & $17,858.97 \pm 288.67$ \\
        Naive MPPI & $10,532.53 \pm 238.07$ & $18,907.35 \pm 321.52$ \\
        DPC        & $\mathbf{132.61 \pm 15.44}$    & $11,182.57 \pm 394.38$\\
        Step-MPPI  & $156.7 \pm 9.98$      & $\mathbf{10,490.32 \pm 229.63}$ \\
        \midrule
        \multicolumn{3}{l}{\textbf{Out-of-Distribution}} \\
        \midrule
        Baseline   & $13,697.57 \pm 326.59$  & $24,145.02 \pm 531.99$ \\
        Naive MPPI & $13,905.52 \pm 318.42$  & $25,844.34 \pm 591.13$ \\
        DPC        & $9,781.82 \pm 1,325.84$  & $21,949.26 \pm 754.86$ \\
        Step-MPPI  & $\mathbf{686.73 \pm 130.32}$    & $\mathbf{17,513.5 \pm 791.2}$\\
        \bottomrule
    \end{tabular*}}
    \vspace{-3mm}
\end{table}

\begin{figure}[t]
    \centering
    \includegraphics[width=0.98\linewidth]{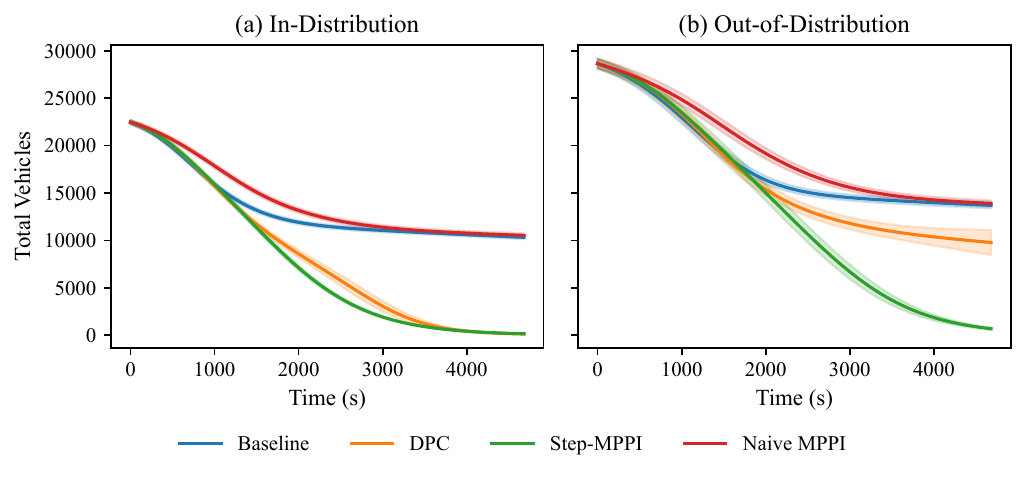}
    \caption{Total network accumulation over time for in-distribution (a) and out-of-distribution (b). %
    }
    \label{fig:nmfd_combined}
    \vspace{-6mm}
\end{figure}

We compare Step-MPPI against a baseline controller, which sets all control inputs to their upper limit, DPC, and MPPI. 
During training, the ``high traffic" cells with indices $(1,5)$, $(0,6)$, $(5,0)$, $(6,1)$ are initialized according to $\mathcal{N}(4000, 100)$ clipped to $[0, 15000]$. 
All other cells are initialized with $\mathcal{N}(150,20)$. 
$50,000$ initial states were sampled for training purposes. The planning horizon used for all methods is $H=40$.
Since Step-MPPI combines sampling-based and learning-based control, it is well-suited to scenarios with distribution shift. 
To evaluate this claim, we compare the methods in an out-of-distribution (OOD) setting, where we replace the initial distribution of the high traffic cells with samples drawn from $\mathcal{N}(5000,200)$.

The results are summarized in Table~\ref{tab:results} and Figure~\ref{fig:nmfd_combined}. 
In the in-distribution (ID) setting, the baseline and naive MPPI both fail to dissipate traffic effectively, leaving over $10{,}000$ vehicles in the network at the end of the horizon. 
Naive MPPI performs slightly worse than the baseline in both metrics, indicating that without a learned prior, its sampling-based optimization struggles in this high-dimensional setting. 
In contrast, DPC and Step-MPPI reduce the final accumulation by nearly two orders of magnitude, to $132.61$ and $156.7$ vehicles, respectively. 
Although DPC achieves a slightly lower final accumulation, Step-MPPI attains a lower TVH (total vehicle hours) of $10{,}490.32$ versus $11{,}182.57$ for DPC, indicating more efficient traffic dissipation over the full horizon. 
Moreover, Step-MPPI exhibits the lowest variance in both metrics, suggesting more consistent performance across initial conditions.

The advantage of Step-MPPI becomes more pronounced in the OOD setting. 
Under distribution shift, DPC degrades dramatically, with final accumulation increasing from $132.61$ to $9{,}781.82$ vehicles and a large standard deviation of $1{,}325.84$, indicating poor generalization beyond the training distribution. 
In contrast, Step-MPPI degrades much more gracefully: its final accumulation increases to $686.73$ vehicles, a $4.4\times$ increase from the ID setting but still over $14\times$ lower than DPC. 
Moreover, its TVH in the OOD setting is $17{,}513.5$, compared to $21{,}949.26$ for DPC, corresponding to a $20\%$ improvement. 
This robustness stems from the hybrid architecture of Step-MPPI, where the learned neural proposal provides an informed initialization and the MPPI refinement step adapts online to unseen states.

Finally, regarding computational cost, Table~\ref{tab:runtime_comparison} reports the average per-step runtimes for all three examples presented in this section. 
Overall, DPC is the fastest learned controller, as it requires only a single forward pass through the policy network. Step-MPPI and Neural Step-MPPI incur additional overhead over DPC due to the single-step MPPI sampling performed at each time step, yet they remain faster than MPPI. 
The higher cost of naive MPPI comes from rolling out all sample sequences over the full planning horizon, resulting in more dynamics evaluations per step than Step-MPPI.
Though Step-MPPI is slower than DPC, it still satisfies real-time requirements (\eg \SI{4.3}{ms} for a robot with a \SI{20}{ms} sampling time).
All controllers are evaluated on an AMD Ryzen 9 7900X with an RTX 4090 GPU. 
The controllers and models are compiled and executed on the GPU using the JAX framework, which leads to the low runtimes.

\begin{table}[t]
\centering
\caption{Runtime comparison (in ms) across different considered examples (Example 1: autonomous vehicles; Example 2: quadrupedal robot; Example 3: traffic network). }
\label{tab:runtime_comparison}
\begin{tabular*}{\linewidth}{@{\extracolsep{\fill}}lccc@{}}
\toprule
\textbf{Method} & \textbf{Example 1} & \textbf{Example 2} & \textbf{Example 3} \\
\midrule
MPPI                & 16.5 & 15.1 & 0.317 \\
DPC                 & 2.9  & 3.8  & 0.075 \\
Step-MPPI           & 7.5  & 4.3  & 0.153 \\
\bottomrule
\end{tabular*}
\vspace{-4mm}
\end{table}

\section{Conclusions}
\label{sec:con}

In this paper, we propose Step-MPPI, a learning-based framework for efficient single-step lookahead MPPI.
By learning the sampling distribution, single-step cost, and MPPI temperature offline, Step-MPPI reduces online execution to an NN prediction followed by a single-step MPPI update while preserving long-horizon objectives during training.
We formulate the MPPI update as a differentiable layer to enable end-to-end training.
Numerical studies on an autonomous vehicle, a quadrupedal robot, and an urban traffic network demonstrate improved computational efficiency, control performance, and robustness over MPPI and DPC.

\section{GenAI Disclosure}
\vspace{-1mm}

ChatGPT (OpenAI) \cite{chatgpt2026} was used to improve the syntax and grammar of several paragraphs in the manuscript. 

\bibliographystyle{IEEEtran}
\balance
\bibliography{IEEEabrv,references}

\end{document}